\documentclass[conference]{IEEEtran}
\usepackage[T1]{fontenc}

\def\argmin{\mathop{\rm \arg\!\min}}
\usepackage{graphicx}
\usepackage[noadjust]{cite}
\usepackage{mcite}
\usepackage{amsfonts,helvet}
\usepackage{fancyhdr}
\usepackage{threeparttable}
\usepackage{epsf,epsfig}
\usepackage{amsthm}
\usepackage{amsmath}
\usepackage{siunitx}
\usepackage{amssymb}
\usepackage{dsfont}
\usepackage{subfigure}
\usepackage{color}
\usepackage{algorithm}
\usepackage{algpseudocode}
\usepackage{algcompatible}
\usepackage{enumerate}
\usepackage{hyperref}
\usepackage{cancel}
\usepackage{bbm}
\usepackage{dsfont}
\usepackage[subnum]{cases}
\usepackage{adjustbox}
\usepackage{blindtext}
\usepackage{tasks}
\usepackage[inline]{enumitem}
\usepackage{scrextend}
\addtokomafont{labelinglabel}{\mdseries }
\usepackage{array}

\makeatletter
\newcommand{\thickhline}{%
    \noalign {\ifnum 0=`}\fi \hrule height 1pt
    \futurelet \reserved@a \@xhline
}
\newcolumntype{"}{@{\hskip\tabcolsep\vrule width 1pt\hskip\tabcolsep}}
\makeatother

\newcommand{\fig}[1]{Fig.\ \ref{#1}}

\newtheorem{proposition}{Proposition}

\newmuskip\pFqmuskip

\IEEEoverridecommandlockouts

\title{ADC Bit Optimization for Spectrum- and Energy-Efficient Millimeter Wave Communications}
\author{Jinseok Choi$^\dagger$, Junmo Sung$^\dagger$, Brian L. Evans$^\dagger$, and Alan Gatherer$^*$ 
\thanks{This research is supported by gift funding from Huawei Technologies.} 
\\
\IEEEauthorblockA{\normalsize{$^\dagger$Wireless Networking and Communications Group, The University of Texas at Austin}\\
Email: \{jinseokchoi89, junmo.sung\}@utexas.edu, bevans@ece.utexas.edu\\
$^*$Systems and Design for Wireless Communications, Huawei USA\\
Email: alan.gatherer@huawei.com}}

\begin{document}
\maketitle

\begin{abstract}

A spectrum- and energy-efficient system is essential for millimeter wave communication systems that require large antenna arrays with power-demanding ADCs.
We propose an ADC bit allocation (BA) algorithm that solves a minimum mean squared quantization error problem under a power constraint.
Unlike existing BA methods that only consider an ADC power constraint, the proposed algorithm regards total receiver power constraint for a hybrid analog-digital beamforming architecture. 
The major challenge is the non-linearities in the minimization problem.
To address this issue, we first convert the problem into a convex optimization problem through real number relaxation and substitution of ADC resolution switching power with constant average switching power.
Then, we derive a closed-form solution by fixing the number of activated radio frequency (RF) chains $M$.
Leveraging the solution, the binary search finds the optimal $M$ and its corresponding optimal solution.
We also provide an off-line training and modeling approach to estimate the average switching power.
Simulation results validate the spectral and energy efficiency of the proposed algorithm.
In particular, existing state-of-the-art digital beamformers can be used in the system in conjunction with the BA algorithm as it makes the quantization error negligible in the low-resolution regime.
\end{abstract}

\section{Introduction}
\label{sec:intro}

Millimeter wave (mmWave) communications significantly increase the data rate by using very large bandwidth \cite{pi2011introduction}.
Large antenna arrays with power-demanding ADCs used to compensate for large pathloss, however, causes problems regarding hardware power and cost.
Previous studies \cite{niu2015survey, heath2016overview} consider two major architectures: a hybrid analog-digital architectures \cite{han2015large} and low-resolution ADC receiver \cite{mo2015capacity}.
The former employs analog beamforming to decrease the number of RF chains to be less than that of antennas, thereby mitigating the burden on digital beamforming, and the latter adopts a small number of quantization bits to reduce ADC power consumption.
Considering a hybrid architecture with low-resolution ADCs \cite{mo2016hybrid, choi2017low}, we propose an algorithm for optimizing ADC bits to realize spectrum- and energy-efficient mmWave communications.


Due to an exponential increase of ADC power with increasing quantization bits, a low-resolution ADC has drawn attention \cite{wang2014multiuser,mo2014channel, mo2016channel, choi2016near}.
Detection and channel estimation techniques compensate for large distortion from using low-resolution ADCs \cite{wang2014multiuser,mo2014channel, mo2016channel, choi2016near} by leveraging large antenna arrays.
For instance, message-passing algorithms achieved better multi-user detection than minimum mean squared error (MMSE) receivers \cite{wang2014multiuser}.
Exploiting the sparsity of mmWave channels in the angular domain, a generalized approximate message-passing algorithm with 1-bit ADCs showed a comparable channel estimation performance as maximum-likelihood (ML) estimator with full-resolution ADCs \cite{mo2014channel}, and it was further extended to wideband mmWave channels in \cite{mo2016channel}.
For Rayleigh MIMO channels, the near ML detector and channel estimator were developed \cite{choi2016near}.
Mixed-resolution architectures were also proposed in which ADC resolution can switch between a 1-bit ADC and high-resolution ADC \cite{liang2016mixed}.
A hybrid architecture with low-resolution ADCs was investigated \cite{mo2016hybrid}, and an ADC bit allocation (BA) algorithm was further developed \cite{choi2017low} with resolution-adaptive ADCs. 
The ADC-switching approach in \cite{liang2016mixed} 
is limited to binary switching between 1-bit and high-resolution ADCs, which are highly power demanding.
Accordingly, the BA approach for mmWave communication in \cite{choi2017low} improved spectral and energy efficiency from conventional systems where ADC resolutions are the same.
In \cite{choi2017low}, the BA algorithm considers only the total ADC power consumption, not the total receiver power consumption.
A more comprehensive BA approach encompassing the entire receiver power consumption should be considered.

In this paper, we propose a BA algorithm minimizing mean squared quantization error (MSQE) of desired signals under constrained total receiver power for a hybrid receiver architecture with resolution-adaptive ADCs.
An analog combiner is designed by using array steering vectors to reduce the number of RF chains.
Then, we formulate a minimum MSQE (MMSQE) problem subject to a total receiver power constraint.
This non-convex minimization problem has non-linear elements such as ADC-switching power and the number of activated RF chains, which makes the solution intractable.
Therefore, we relax the problem formulation to make it convex and derive a near optimal low-complexity solution.
First, we relax the integer problem to a real number problem in order  to use Karush-Kuhn-Tucker (KKT) conditions. 
We, then, replace the switching power with a constant average switching power, of which we further provide an off-line approach to estimate.
To resolve non-linearity of the number of activated RF chains, we derive a closed-form solution by fixing the number of activated RF chains $M$, and find the optimal $M$ with the smallest MSQE through binary search.
Simulation results validate the performance of the proposed BA algorithm.

\begin{figure}[t]\centering
\includegraphics[scale = 0.3]{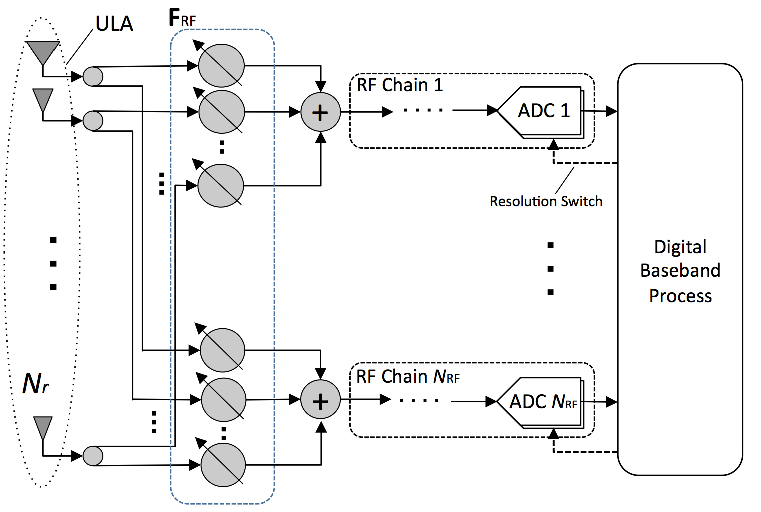}
\caption{Hybrid massive MIMO reciever with resolution-adaptive ADCs.} 
\label{fig:system}
\end{figure}

\section{System Model}
\label{sec:sys_model}

\subsection{Signal and Channel Models}
We consider single-cell MIMO uplink networks in which a base station (BS) with $N_r$ antennas serves $N_{u}$ users with a single antenna ($N_r \gg N_u$).
The hybrid architecture with low-resolution ADCs is adopted at the BS, which employs uniform linear array (ULA) and analog combiners.
$N_{\rm RF}$ RF chains are connected to $N_{\rm RF}$ pairs of ADCs.
Without loss of generality, we assume that the number of RF chains is less than the number of antennas ($N_{\rm RF} < N_r$), reducing the power consumption and complexity at the BS. 
As shown in Fig. \ref{fig:system}, resolution-adaptive ADCs such as flash ADCs  \cite{yoo2002power}.



We assume mmWave channels in which limited $L$ scatterings ($L \ll N_r$) contribute to $L$ propagation paths \cite{el2014spatially}.
Then, the $k^{th}$ user channel is represented as
\begin{align}
\label{eq:channel_geo}
{\bf h}_k = \sqrt{\gamma_k}\sum_{\ell = 1}^{L}\omega^k_\ell {\bf a}(\theta^k_\ell) \in \mathbb{C}^{N_r}
\end{align}
where $\gamma_k$ denotes the large scale fading gain between the BS and $k^{th}$ user and includes geometric attenuation, shadow fading, and noise power.
The vector ${\bf a}(\theta^k_{\ell})$ and $\omega^k_{\ell}$ represent a BS antenna array steering vector with the azimuth angle of arrival (AoA) of the $\ell^{th}$ path for the $k^{th}$ user $\theta^k_{\ell} \in [-\pi/2,\pi/2]$, and complex gain of the $\ell^{th}$ path for the $k^{th}$ user, respectively. 
We consider that $\omega^k_{\ell}$ is an independent and identically distributed (IID) complex Gaussian random variable $\omega^k_{\ell} \sim \mathcal{CN}(0, 1)$.

For ULA, the array steering vector ${\bf a}(\theta)$ becomes a vector whose $n^{th}$ entry is $\frac{1}{\sqrt{N_r}}e^{-j \frac{2\pi (n-1)d}{\lambda}\sin(\theta)}$
where $\lambda$ is the wave length and $d$ is the distance between antenna elements.
Assuming uniformly-spaced spatial angle $\frac{d}{\lambda}\sin(\theta_i) = (i-1)/N_r$, the matrix of array steering vectors
$\mathbf{A}=\big[{\bf a}(\theta_1),\cdots,{\bf a}(\theta_{N_r})\big]$
becomes a unitary discrete Fourier transform matrix. 
Then, the channel vector ${\bf h}_k$ \eqref{eq:channel_geo} can be modeled as \cite{sayeed2002deconstructing}
\begin{align} 
\nonumber
{\bf h}_k &= {\bf A}{ \tilde{ \bf h}}_{{\rm b},k}=\sum_{i = 1}^{N_r}{ \tilde h}_{{\rm b},(i,k)}\,{\bf a}(\theta_i)
\end{align}
where $\tilde{\bf h}_{{\rm b},k} \in \mathbb{C}^{N_r}$ denotes the beamspace channel that includes both the $L$ complex gains $\sim \mathcal{CN}(0,1)$ and the large scale fading gain $\sqrt{\gamma_k}$.
The beamspace channel matrix becomes $\tilde{\bf H}_{\rm b} = [\tilde{\bf h}_{{\rm b},1},\cdots,  \tilde{\bf h}_{{\rm b},N_u}]=\tilde{\bf  G}{\bf D}^{1/2}$ where $\tilde{\bf  G} \in \mathbb{C}^{N_r \times N_u} $ is the sparse matrix that only contains complex path gains and ${\bf D} = {\rm diag}(\gamma_1,\cdots,\gamma_{N_u})$ is the diagonal matrix of large scale fading gains.
Accrodingly, the beamspace channel vector becomes $ \tilde{\bf h}_{{\rm b},k} = \sqrt{\gamma_k}\tilde {\bf g}_k$.
Finally, the channel matrix ${\bf H} = [{\bf h}_1, \cdots,{\bf h}_{N_u}]$ is given as
\begin{align}
 \label{eq:channel}
{\bf H} =  {\bf A} \tilde {\bf H}_{\rm b} = {\bf A}\tilde {\bf G }{\bf D}^{1/2}.
\end{align}

The received baseband analog signal ${\bf r} \in \mathbb{C}^{N_r}$ with a narrowband channel assumption is expressed as
\begin{align}
	\label{eq:rx_signal}
	{\bf r} = \sqrt{p_u}{\bf Hs} + \tilde {\bf n}
\end{align} 
where $p_u$, ${\bf s}  \in \mathbb{C}^{N_u}$, and $ \tilde {\bf n} \in \mathbb{C}^{N_r}$ represent the transmit power, the symbol vector of $N_u$ users, and the additive white Gaussian noise, respectively. 
We assume Gaussian signaling for $\bf s$ and a normalized noise variance $\tilde {\bf n}  \sim \mathcal{CN}(\mathbf{0}, \mathbf{I}_{N_r})$. 
We consider that the channel $\bf H$ is perfectly known at the BS.

An analog beamformer ${\bf F_{\rm RF}}  \in \mathbb{C}^{N_r \times N_{\rm RF} }$ is applied to ${\bf r}$ in \eqref{eq:rx_signal} and its elements are constrained to have the equal norm of $1/\sqrt{N_r}$.
We assume that ${\bf F_{\rm RF}}$  is composed of $N_{\rm RF}$ array response vectors corresponding to the $N_{\rm RF}$ strongest propagation paths \cite{el2012capacity}: ${\bf F_{\rm RF}} = {\bf A}_{\rm RF}$ where ${\bf  A}_{\rm RF}$ is a $N_r \times N_{\rm RF}$ sub-matrix of $\bf A$.
We consider that $ {\bf A}_{\rm RF}$ embraces all propagation paths from $N_u$ users. 
Then, the received signal becomes
\begin{align} 
\label{eq:y}
{\bf {y}}  = \sqrt{p_u}{\bf  A}_{\rm RF}^H{\bf Hs} +  {\bf A}_{\rm RF}^H \tilde  {\bf n} =\sqrt{p_u} {\bf  H_{\rm b}}{\bf s} + {\bf n} 
\end{align}
where $  {\bf n} =  {\bf A}_{\rm RF}^H \tilde{\bf n} \sim \mathcal{CN}(\mathbf{0}, \mathbf{I}_{N_{\rm RF}})$ as $\bf A$ is unitary, and $(\cdot)^H$ denotes the conjugate transpose. 
Note that ${\bf y}$ is a $N_{\rm RF}\times 1$ vector, and ${\bf H}_{\rm b}$ is a $N_{\rm RF} \times N_u$ beamspace channel matrix that can be decomposed as
${\bf H_{\rm b} = GD}^{1/2}$
where ${\bf G }\in \mathbb{C}^{N_{\rm RF} \times N_u} $ is the sub-matrix of $ \tilde {\bf G}$ that corresponds to ${\bf A}_{\rm RF}$.

\subsection{Quantization Model}
\label{subsec:quantization}

Each output $y_i$ is quantized at the $i^{th}$ pair of ADCs as shown in \fig{fig:system}.
Each ADC in the pair quantizes either a real or imaginary part of $y_i$ with the same number of quantization bits $b_i$.
Adopting the additive quantization noise model (AQNM) which provides reasonable accuracy in low and medium SNR ranges \cite{orhan2015low}, we have the linear quantization expression of $\bf y$
\begin{align} 
	\nonumber
	\mathbf{y}_{\rm q}&=\mathcal{ Q}(\mathbf{ y}) = \mathbf{W}_{\alpha} \,\mathbf{ y}+ \mathbf{n}_{\rm q}
	\\ \label{eq:AQNM2}
	& = \sqrt{p_u}{\bf W_\alpha  H_{\rm b}}{\bf s + W_\alpha} {\bf n} + {\bf n}_{\rm q}
\end{align} 
where $\mathcal{Q}(\cdot)$ is the element-wise quantizer function for each real and imaginary part,
and $\mathbf{W}_\alpha$ is the diagonal quantization gain matrix $\mathbf{W}_\alpha =  {\rm diag}(\alpha_1,\cdots, \alpha_{N_{\rm RF}})$.
The quantization gain is defined as $\alpha_i = 1- \beta_i$ where $\beta_i = \frac{\mathbb{E}[|{y}_i - {y}_{{\rm q}i}|^2]}{\mathbb{E}[|{y}_i|^2]} $.
Note that $\mathbb{E}\big[|{y}_i - {y}_{{\rm q}i}|^2\big]$ is the MSQE of $y_i$.
Under the assumptions of a scalar MMSE quantizer and Gaussian signaling, Table 1 in \cite{fan2015uplink} lists the values of $\beta_i$ for $b_i \leq 5$.
For $b_i > 5$, $\beta_i$ can be modeled as $\beta_i = \frac{\pi\sqrt{3}}{2} 2^{-2b_i}$ \cite{orhan2015low}. 
The quantization noise $\mathbf{n}_{\rm q}$ is an additive noise that is uncorrelated with $\bf  y$.
The noise follows the zero-mean complex Gaussian distribution, and the covariance matrix for a fixed channel realization $ {\bf H}_{\rm b}$ is
\begin{align}
\label{eq:cov2}
\mathbf{R}_{\mathbf{n}_{\rm q}}= {\bf W}_\alpha {\bf W}_\beta \,{\rm diag}(p_u{\bf H_{\rm b}}{\bf H}_{\rm b}^H + {\mathbf{I}_{N_{\rm RF}}})
\end{align}
where $\mathbf{W}_\beta =  {\rm diag}(\beta_1,\cdots, \beta_{N_{\rm RF}})$, and ${\rm diag}(p_u{\bf H_{\rm b}}{\bf H}_{\rm b}^H + {\mathbf{I}_{N_{\rm RF}}})$ takes the diagonal entries of $ p_u{\bf H_{\rm b}}{\bf H}_{\rm b}^H + {\mathbf{I}_{N_{\rm RF}}}$.



\section{ADC Bit Optimization}
\label{sec:Problem}

In this section, we develop a BA algorithm that provides near optimal ADC bit configurations which minimize the quantization error under constrained total power consumption.

\subsection{Problem Formulation}
\label{ssec:formulation}

We first formulate a minimization problem by adopting the MSQE as a distortion measure.
We remark that the MSQE of desired signals, ${\bf x} =\sqrt{p_u} {\bf  H_{\rm b}}\bf s$, rather than received signals, $\bf y$ \eqref{eq:y}, is considered since the BA method on the desired signals is shown to be robust to noise \cite{choi2017low}.
Applying the AQNM to the desired signals, 
the MSQE of the $i^{th}$ desired signal $x_i $ with the non-linear scalar MMSE quantizer for $b_i > 5$ is \cite{orhan2015low}
\begin{align}
\label{eq:msqe_rev}
{\mathcal{E}}_{i}(b_i) 
&= \mathbb{E}\Big[|{x}_i - {x}_{{\rm q}i}|^2\Big]= \frac{\pi\sqrt{3}}{2}\sigma_{x_i}^2\,2^{-2b_i}
\end{align}
where $x_{{\rm q},i}$ is the quantized ${ x}_i$ and $\sigma^2_{x_i} = p_u\|[{\bf H}_{\rm b}]_{i,:}\|^2$. 
Note that $\sigma^2_{x_i}$ is known at the BS as we assume that the BS knows the channel \eqref{eq:channel}.
We consider \eqref{eq:msqe_rev} to hold for any $b_i$ when formulating a minimization problem.  

To calculate receiver power consumption, we adopt power consumption models for ADC quantization and resolution switching.
The ADC quantization power consumption is \cite{orhan2015low} 
\begin{align}
	\label{eq:ADCpower}
	P_{\rm ADC}(b)= c\, f_s\,2^{b}
\end{align}
where $c$ is the energy consumption per conversion step (conv-step), called Walden's figure-of-merit, $f_s$ is the Nyquist sampling rate, and $b$ is the number of quantization bits. 
\eqref{eq:ADCpower} implicitly implies that $P_{\rm ADC}(b) = 0$ for $b = 0$.
\textcolor{black}{Note that additional power can be consumed due to switching resolution, and the resolution switching power consumption is modeled as} \cite{choi2017low}
\begin{align}
	\label{eq:Psw_model}
	P_{\rm SW}(b) = c_{\rm sw}\big|2^{b} - 2^{b^{\rm p}}\big|
\end{align} 
where $b^{\rm p}$ is the number of previous ADC bits. 
The power consumption per conversion step $c_{\rm sw}$ has different values for increasing $(b > b^{\rm p})$ and decreasing resolution $(b < b^{\rm p})$ cases.
\textcolor{black}{This model \eqref{eq:Psw_model} implies that the switching power is proportional to the number of comparators that are turned on or off.}
$P_{\rm SW}(b)$ becomes zero when no change in resolution ($b = b^{\rm p}$).

Let $P_{\rm LNA}$, $P_{\rm PS}$, $P_{\rm RFchain}$, and $P_{\rm BB}$ indicate the power consumption in the low-noise amplifier, phase shifter, RF chain, and baseband processor, respectively.
We also consider that RF processes associated with deactivated ADCs (0-bit ADCs) after bit allocation are turned off so that no power consumption occurs in the RF processes.
Accordingly, the receiver power consumption of the system in Fig. \ref{fig:system} is
\begin{align}
	\nonumber
	P_{\rm tot} =& N_rP_{\rm LNA} + N_{\rm act}(N_rP_{\rm PS} + P_{\rm RFchain}) \\
	\label{eq:total_power}
	&+ 2\sum_{i=1}^{N_{\rm RF}}\Big(P_{\rm ADC}(b_i) + P_{\rm SW}(b_i)\Big)+ P_{\rm BB} 
\end{align}
where $N_{\rm act}$ is the number of activated RF chains. 

Using the MSQE of desired signals \eqref{eq:msqe_rev} and the receiver power \eqref{eq:total_power}, the minimum MSQE problem under a power constraint $p$ is formulated as
\begin{gather}
	\label{eq:MMSQE_problem}
	\hat {\mathbf{{b}}}^{\rm } 
	= \argmin_{\mathbf{b}\in\mathbb{Z_+}^{N_{\rm RF}}
	} \sum_{i=1}^{N_{\rm RF}}\mathcal{E}_{i}(b_i)
	\quad \text{s.t.} \quad P_{\rm tot} \leq p
\end{gather}
where $\mathbb{Z_+}$ represents the set of non-negative integers.
Note that the integer variables and non-linear functions that are the number of activated RF chains $N_{\rm act}$ and ADC-switching power consumption $P_{\rm SW}(b_i)$ require an exhaustive search to solve \eqref{eq:MMSQE_problem}.
To derive a near optimal solution with low complexity, we develop a BA algorithm in the next subsection.

\subsection{ADC Bit Optimization Algorithm}
\label{ssec:BA}

We modify the problem \eqref{eq:MMSQE_problem} with proper relaxations.
We first relax the domain of the variables to real numbers (${\bf b} \in \mathbb{Z}^{N_{\rm RF}}_+  \to {\bf b} \in \mathbb{R}^{N_{\rm RF}} $), and then, relax the ADC power consumption model \eqref{eq:ADCpower} to be differentiable for any value of the quantization bits ($P_{\rm ADC}(b) = cf_s2^b$, $ b \in \mathbb{R}$).
Despite the real number relaxation, it is still difficult to solve \eqref{eq:MMSQE_problem} due to non-linear factors in $P_{\rm tot}$: the switching power consumption $P_{\rm SW}(b_i)$ and the number of activated RF chains $N_{\rm act}$.

To resolve this challenge, we replace $P_{\rm SW}(b_i)$ with a constant $\bar P_{\rm SW}$, which is the average switching power consumption when applying the optimal solution of \eqref{eq:MMSQE_problem}.
For a given system environment, $\bar P_{\rm SW}$ depends only on the power constraint $p$, i.e., $\bar P_{\rm SW}$ is a function of $p$.
Since the amount of switching power consumption directly affects the power available for quantization, 
an accurate estimation of $\bar P_{\rm SW}$ is required.
Accordingly, we propose an estimation approach to $\bar P_{\rm SW}$ by using off-line training and modeling in Section \ref{ssec:Psw}.

Finally, the non-linearity of $N_{\rm act}$ can be addressed by following the general steps:
\begin{labeling}{Step 3.}
\item [Step 1.] Sort variances $\sigma^2_{x_i}$ to be $\sigma^2_{x_1} \geq \sigma^2_{x_2} \geq \cdots \geq \sigma^2_{x_{N_{\rm RF}}}$
\item [Step 2.] Derive a solution $\hat {\bf b}$ for \eqref{eq:MMSQE_problem} when the first $M$ RF chains are considered to be used ($N_{\rm act} = M$)
\item [Step 3.] Find the optimal value of $M$ that gives the smallest error $\sum_{i = 1}^{N_{\rm RF}} \mathcal{E}_i(\hat b_i)$ with its corresponding solution
\end{labeling}
Let $M_{\rm opt}$ represent the optimal value of $M$, then the solution with $N_{\rm act} = M_{\rm opt}$ is the near optimal solution for \eqref{eq:MMSQE_problem}.
Intuitively, Step 1 is designed from the fact that ADCs with larger channel gains are likely to be assigned with more quantization bits \cite{choi2016adc, choi2017low}.
For Step 2, we can derive a closed-form solution of the relaxed minimization problem with $N_{\rm act} = M$ by solving the KKT conditions.
Binary search can be applied in Step 3 as the total MSQE  $\sum_{i = 1}^{N_{\rm RF}} \mathcal{E}_i(b_i)$ decreases to the point of $M_{\rm opt}$ and increases beyond $M_{\rm opt}$.
Hence, the search complexity reduces to $\mathcal{O}(\log_2 N_{\rm RF})$. 


Assuming $\sigma^2_{x_1} \geq \sigma^2_{x_2} \geq \cdots \geq \sigma^2_{x_{N_{\rm RF}}}$ from Step 1, we now derive a solution of the relaxed problem at the $s^{th}$ stage of binary search where the first $M_s$ RF chains are considered to be used ($N_{\rm act} = M_s$).
Consequently, we deactivate the ADCs corresponding to the rest of the RF chains, i.e., $b_i = 0$ for $i = M_s+1, \cdots, N_{\rm RF}$.
\vspace{0.047 in}
Recall that we previously relaxed (i) ${\bf b} \in \mathbb{Z}^{N_{\rm RF}}_+  \to {\bf b} \in \mathbb{R}^{N_{\rm RF}} $ with $P_{\rm ADC}(b) = cf_s2^b$ for $b \in\mathbb{R}$ and (ii) $P_{\rm SW}(b_i) \to \bar P_{\rm SW}$.
Then, the relaxed minimization problem at the $s^{th}$ stage of binary search is formulated as
\begin{gather}
	\label{eq:Relaxed_Problem}
	{\mathbf{b}}^{\rm s}
	= \argmin_{\mathbf{b}\in\mathbb{R}^{M_{s}}} \sum_{i=1}^{M_s}\mathcal{E}_{i}(b_i)
	\quad \text{s.t.} \quad 2\sum_{i=1}^{M_{s}} P_{\rm ADC}(b_i) \leq \tilde p
\end{gather}
where 
\small
\begin{align}
\nonumber
\tilde p = &\,p-N_rP_{\rm LNA} - 2{N_{\rm RF}}\bar P_{\rm SW} - P_{\rm BB} - M_s(N_rP_{\rm PS}+P_{\rm RFchain}). 
\end{align}
\normalsize
Note that the power constraint for ADC quantization $\tilde p$ cannot be negative, which restricts the maximum number of activated RF chains.
Moreover, we can exclude the RF chains with zero channel gain -- $M_s$ is less than or equal to the total number of non-zero $\sigma^2_{x_i}$.
Therefore, the number of the considered RF chains $M_s$ needs to satisfy the following condition:
\begin{align}
\label{eq:Ms_condition}
M_s \leq \min\big (M_{\rm max}, N_{\rm RF}\big)
\end{align}
with
\footnotesize
\begin{align*}
&M_{\rm max}
= \min\bigg(\bigg \lfloor \frac{p-N_rP_{\rm LNA} - 2{N_{\rm RF}}\bar P_{\rm SW} - P_{\rm BB}}{N_rP_{\rm PS}+P_{\rm RFchain} }\bigg \rfloor, \sum_{i=1}^{N_{\rm RF}} \mathds{1}_{\{\sigma^2_{x_i} \neq 0\}} \bigg).
\end{align*}
\normalsize
where $\mathds{1}_{\{\cdot\}}$ is the indicator function.
For \eqref{eq:Relaxed_Problem} with $M_s$ under \eqref{eq:Ms_condition}, the closed-form solution is derived in Proposition \ref{prop:BA}.  
\begin{proposition} 
\label{prop:BA}
With $M_s$ under \eqref{eq:Ms_condition}, the optimal number of quantization bits for the convex minimization problem \eqref{eq:Relaxed_Problem} is 
\begin{align}
\nonumber
b^{\rm s}_i = &\log_2 \frac{\tilde p}{2c\,f_s} + \log_2\left(\frac{\|[{\bf H}_{\rm b}]_{i,:}\|^{\frac{2}{3}}}{\sum_{j = 1}^{M_{s}}\|[{\bf H}_{\rm b}]_{j,:}\|^{\frac{2}{3}}}\right),\ i = 1,\cdots, M_{s}.
\end{align}
\end{proposition}
\begin{proof}
The problem \eqref{eq:Relaxed_Problem} is a convex minimization problem and achieves a global minimum when the power constraint holds equal.
Let $w_i = \sigma^2_{x_i}$, $z_i = 2^{b_i}$, $f({\bf z}) = \sum_{i=1}^{M_s} w_iz_i^{-2}$, and $g({\bf z}) =\sum_{i=1}^{M_s}z_i - \frac{\tilde p}{2c\,f_s}$, then \eqref{eq:Relaxed_Problem} reduces to
\begin{gather}
\label{eq:eq_problem}
\hat {\bf z} = \argmin_{\mathbf{z}\in \mathbb{R}^{M_s}} {f({\bf z})}
\quad \text{s.t.}\quad g({\bf z}) = 0,\  {\bf z} > \bf 0
\end{gather}
where $\bf 0$ and the inequality in ${\bf z} > \bf 0$ denote a zero vector and element-wise inequality, respectively.
The global minimizer $\hat {\bf z}$ for the problem \eqref{eq:eq_problem} satisfies the KKT conditions: 
\begin{enumerate*}
\item  $\nabla f(\hat{\bf z}) + \mu \nabla g(\hat{\bf z}) = \bf 0$, \label{cond1}
\item $g(\hat{\bf z}) = 0$, and \label{cond2}
\item  $\hat{\bf z} > \bf 0$,\label{cond3}
\end{enumerate*}
where $\mu \in \mathbb{R}$ is a Lagrange multiplier.
Using the KKT condition \ref{cond1} which gives $-2w_i\hat z_i^{-3} + \mu = 0$ and condition \ref{cond2}, we can derive 
\begin{align}
\label{eq:zi}
\hat z_i = \tilde p\, w_i^{1/3}\Big/\Big(2c\,f_s\sum_{j=1}^{M_s}w_j^{1/3}\Big) > 0,\quad i = 1,\cdots, M_s.
\end{align}
The derived solution $\hat z_i$ in \eqref{eq:zi} meets the KKT conditions. 
Therefore, using the definitions of $w_i$ and $z_i$, we obtain the final solution in Proposition \ref{prop:BA}.
\end{proof}

Proposition~\ref{prop:BA} illustrates that increasing the number of considered RF chain $M_s$ in $\tilde p$ decreases that of ADC bits $b^{\rm s}_i$ while it increases that of non-zero ADC bits activating more ADCs. 
Due to such trade-off between the resolution of ADCs and the number of activated ADCs, there exists an optimal number of RF chains $M_{\rm opt}$ gives a minimum total quantization error with
\begin{align}
\label{eq:binary_search}
\hat {\bf b}^{\intercal} = 
\begin{bmatrix}
\max({\bf b}^{\rm s},{\bf 0})^{\intercal}, 
{\bf 0}^{\intercal} 
\end{bmatrix} 
\in \mathbb{R}^{1\times N_{\rm RF}}
\end{align}
where $\max(\cdot)$ is the element-wise max operator and $(\cdot)^{\intercal}$ denotes the transpose.
After finding the optimal solution $\hat {\bf b}$ corresponding to $M_{\rm opt}$, we map the real number solution $\hat {\bf b}$ to its nearest integers.
Algorithm \ref{algo:ADC} shows the proposed BA algorithm for ADC bit optimization.

\begin{algorithm}[b]
\caption{BA algorithm for ADC bit optimization}
\label{algo:ADC}
\begin{enumerate}
\item Set a power constraint $p$
\item Sort the variance $\sigma^2_{x_i}$ to be $\sigma^2_{x_1} \geq \sigma^2_{x_2} \geq \cdots \geq \sigma^2_{x_{N_{\rm RF}}}$
\item Compute $M_{\rm max}$ \eqref{eq:Ms_condition} and $ M^\dagger_{\rm max} =\min(M_{\rm max}, N_{\rm RF})$
\item Set $\mathbb S = \{1, \ldots, M^\dagger_{\rm max}\}$ 
\item Binary search at stage $s$ with $M_s \in \mathbb{S}$
\begin{enumerate}
	\item $M_{s}^{\rm L} =\max(1,M_s-1)$, $M_{s}^{\rm R} =\min( M^\dagger_{\rm max}, M_s+1)$
	\item For $M_{\rm s}^{\rm L}, M_s, M_{\rm s}^{\rm R}$
	\begin{enumerate}
		\item compute ${\bf b}^{\rm s}$ using Proposition \ref{prop:BA}, then
		\item compute $\hat {\bf b} $ in \eqref{eq:binary_search}
	\end{enumerate}
	\item Compare total MSQE $\sum_{i =1}^{N_{\rm RF}}\mathcal{E}_i(\hat b_i)$ for $M_{ s}^{\rm L}, M_s, M_{s}^{\rm R}$
	\item {\bf if} $\nexists$ $M^{\rm L}_s$ or $M^{\rm R}_s$ with smaller total MSQE than $M_s$
	\begin{enumerate}
		\item map $\hat{\bf b}$ of $M_s$ to nearest integer numbers, then
		\item {\bf return} the integer $\hat{\bf b}$
	\end{enumerate}
	\item {\bf else} go to the smaller half.
\end{enumerate}
\end{enumerate}
\end{algorithm}

\subsection{Training and Modeling of Average Switching Power }
\label{ssec:Psw}

In this subsection, we propose a training and modeling approach for the estimation of average ADC-switching power consumption $\bar P_{\rm SW}$, which can be computed off-line.
Note that if we consider smaller $\bar P_{\rm SW}$, the available power for quantization becomes larger, which can increase the actual average switching power consumption and vice versa.
Based on this intuition, we find the estimated average switching power consumption that has the minimum absolute difference from actual average switching power consumption.
For a given power constraint $p$, the training steps follow:
\begin{labeling}{Step 4.}
\item [Step 1.] Set an estimated average switching power $\bar P_{\rm est}$
\item [Step 2.] Perform Algorithm \ref{algo:ADC} over different channel realizations and calculate actual average switching power consumption $\bar P_{\rm act}$ using \eqref{eq:Psw_model}
\item[Step 3.] 
Repeat Step 1 and 2 for different values of $\bar P_{\rm est}$
\item[Step 4.] Find $\bar P_{\rm est}^*$ with the smallest absolute difference from the corresponding $\bar P_{\rm act}^*$, and set $T_p = \bar P_{\rm est}^*$
\end{labeling}
Accordingly, the training result $T_p$ is the best estimate of average switching power consumption for power constraint $p$. 
From the model in \eqref{eq:Psw_model}, we can determine a reasonable range of average switching power consumption for $\bar P_{\rm est}$ in Step 1.

Once we obtain $T_p$ for different power constraints $p$, we further perform least-squares polynomial (LSP) fitting to model $\bar P_{\rm SW}$ as a function of a power constraint $p$.
Note that this training and modeling can be computed off-line, and thus does not increase the complexity of the proposed ADC bit optimization algorithm. 
In Section \ref{sec:Num}, we provide the training data $T_p$ and its LSP fitting result.
Then, the average switching power function obtained from the LSP fit is used when evaluating the performance of the proposed BA algorithm.

\begin{figure}[t]\centering
\includegraphics[scale = 0.33]{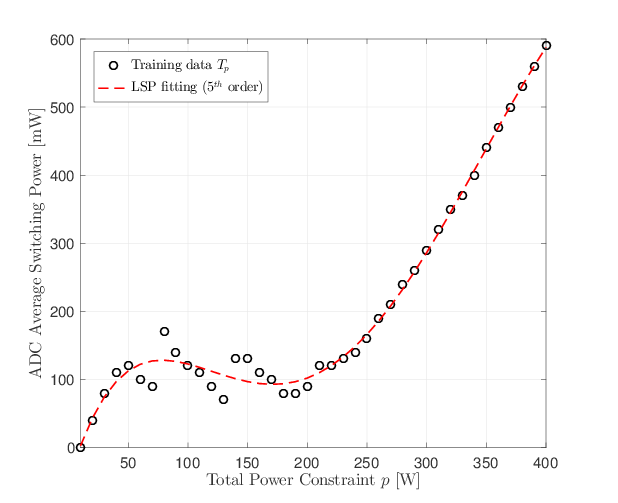}
\caption{Training data $T_p$ and least-squares polynomial fitting curve for $256$ BS antennas, $128$ RF chains, $10$ users and $13$ propagation paths per user.} 
\label{fig:FittingCurve}
\vspace{-0.8 em}
\end{figure}

\section{Numerical Results}
\label{sec:Num}



We assume that users are randomly distributed over a single cell within a 200 $m$ radius, and that 
the minimum distance between the BS and users is $30$ $m$.
Considering a $28$ GHz carrier frequency, we adopt the mmWave pathloss model \cite{akdeniz2014millimeter}
\begin{align}
\nonumber
PL(d_k)\text{ [dB]} = \alpha_{\rm pl} + \beta_{\rm pl} 10\log_{10}d_k + \chi,\quad \chi\sim\mathcal{N}(0,\sigma^2_{\rm s}) 
\end{align}
where $d_k$ [$m$] is the distance between the BS and $k^{th}$ user, $\alpha_{\rm pl}  = 72$ dB, and $\beta_{\rm pl}  = 2.92$.
The variance $\sigma_s^2 = 8.7$ dB and is the lognormal shadowing variance. 
The noise power with transmission bandwidth $W$ and noise figure $n_f$ is  
\begin{align}
P_{\rm noise}\text{ [dBm]} = -174 + 10\log_{10} W + n_f.
\end{align}
We assume $W = 1$ GHz so as $f_s = 1$ GHz, and $n_f=5$ dB.
Due to the normalized noise variance in the system model \eqref{eq:rx_signal}, the large scale fading gain is given as
\begin{align}
\label{eq:large_scale_fading_gain}
 {\gamma}_{k,\text{dB}} \text{ [dB]} = -(PL(d_k)   + P_{\rm noise})
\end{align}
We consider $N_r = 256$ BS antennas, $N_u = 10$ users, and $p_u = 20$ dBm transmit power. 
We further assume that the number of RF chains $N_{\rm RF}$ is half of $N_r$ and the number of captured propagation paths for each user, $L$, is about $5 \%$ of $N_r$, i.e., $N_{\rm RF}= 128$ and $L = 13$ so that $L = \lceil \frac{N_{\rm RF}}{N_u}\rceil$.

The training data $T_p$ of the average ADC-switching power is illustrated in Fig. \ref{fig:FittingCurve} with its LSP fitting curve. 
In Fig. \ref{fig:FittingCurve}, $T_p$ shows a general increasing trend as the total power constraint $p$ in \eqref{eq:MMSQE_problem} increases. 
For small $p$, however, there exist power fluctuations.
This is mainly because the optimal number of the activated RF chains $M_{\rm opt}$ changes dynamically as $p$ varies in the small $p$ regime. 
On the other hand, when $p$ reaches certain level, $M_{\rm opt}$ changes slowly because most of the RF chains with large channel gains are already activated -- it is beneficial to increase ADC resolution of currently activated RF chains than to activate more RF chains with small channel gains under a limited power budget.
As shown in Fig. \ref{fig:FittingCurve}, the LSP fitting provides a strong matching with $T_p$ using fifth-order polynomials. 
Note that polynomial coefficients for all possible $N_u^{\rm max} \times L^{\rm max}$ settings can be readily stored in a look-up table through off-line training, where $N_u^{\rm max}$ and $L^{\rm max}$ denote the maximum number of users and propagation path per user, respectively.
In the following subsection, we use the fifth-order LSP fitting model to determine $\bar P_{\rm SW}$ in \eqref{eq:Relaxed_Problem} for a given total power constraint $p$.

\subsection{Spectral Efficiency and Energy Efficiency}
\label{subsec:EE}

\begin{figure}[!t]
\centering
$\begin{array}{c}
{\resizebox{0.85\columnwidth}{!}
{\includegraphics{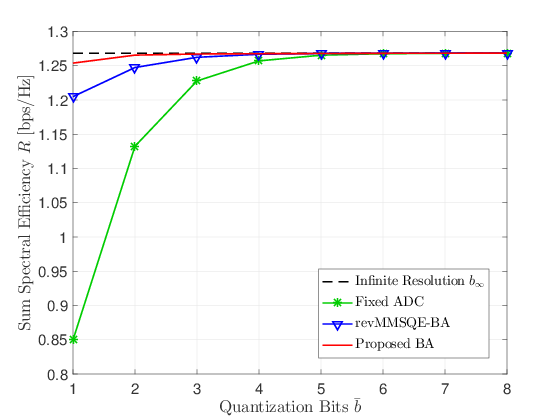}}
} \\ 
 \mbox{(a)} \\
{\resizebox{.85\columnwidth}{!}
{\includegraphics{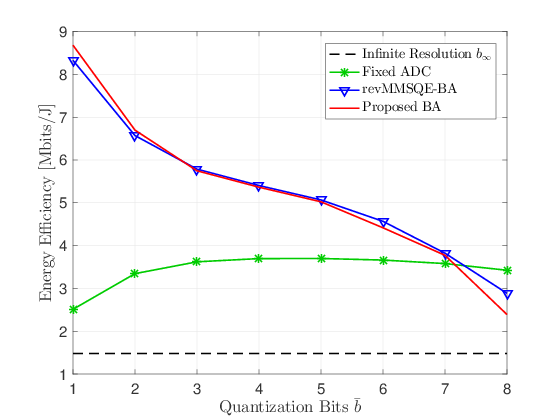}}
}\\
\mbox{(b)}
\end{array}$
\caption{
Uplink (a) sum spectral efficiency and (b) energy efficiency for $256$ BS antennas, $128$ RF chains, $10$ users, and $13$ propagation paths per user .} 
\label{fig:REE}
\vspace{-0.8 em}
\end{figure}

We adopt sum spectral efficiency and energy efficiency.
The uplink sum spectral efficiency is $R = \sum_{k = 1}^{N_u} R_k$ $[\rm bps/Hz]$ where $R_k$ is the ergodic spectral efficiency of the $k^{th}$ user.
Applying maximum ratio combining ${\bf F}_{\rm mrc} = {\bf W}_\alpha {\bf H}_{\rm b}$ to \eqref{eq:AQNM2},
\begin{align}
\label{eq:AchievableRate}
R_{k} = \mathbb{E} \Bigg[ \log_2\bigg(1+\frac{p_u\gamma_k|\pmb \alpha^H{\bf v}_{k}|^2}{{\rm UI}_k + {\rm N}_k + {\rm QN}_k }
\bigg)\Bigg]
\end{align}
where $\pmb \alpha = [\alpha^2_1, \cdots, \alpha^2_{N_r}]^\intercal$, ${\bf v}_{k} = \left[|g_{1,{k}}|^2,\cdots,|g_{N_r,{k}}|^2\right]^\intercal$, ${\rm UI}_k = p_u\sum_{\substack{m = 1\\ m \neq k}}^{N_u}\gamma_m|{\bf g}_{k}^H {\bf W}^2_\alpha {\bf g}_m|^2$, ${\rm N}_k = {\bf g}_{k}^H{\bf W}^4_{\alpha}{\bf g}_{k}$, and ${\rm QN}_k = {\bf g}_{k}^H{\bf W}^H_\alpha \mathbf{R}_{\mathbf{n}_{\rm q}}{\bf W}_\alpha {\bf g}_{k}$.
Then, energy efficiency is \cite{mo2016hybrid}
\begin{align}
\eta_{\rm EE} = {R\,W}/{P_{\rm tot}} \quad {\rm bits/Joule}.
\end{align}
We assume $P_{\rm LNA} = 20$ mW, $P_{\rm PS} = 10$ mW, $P_{\rm RFchain} = 40$ mW, and $P_{\rm BB} = 200$ mW \cite{mo2016hybrid}. 
We consider $c = 494$ fJ/conv-step \cite{orhan2015low} for the quantization power in \eqref{eq:ADCpower}. 
In terms of the switching power in \eqref{eq:Psw_model}, we assume $c_{\rm sw} = 3.47 $ (or $0.94$) mW/conv-step when $b_i > b^{\rm p}_i$ (or $ b_i < b^{\rm p}_i$) \cite{nahata2004high}.

We evaluate four different approaches: (1) infinite resolution $b_\infty = 12$ bits, (2) conventional fixed ADC,
(3) revMMSQE-BA in \cite{choi2017low}, and (4) proposed-BA algorithm.
We enforce harsh simulation constraints which increase switching power consumptions: resolution switching occurs at every symbol transmission by assuming channel coherence time equal to symbol duration, and channel realizations from one coherence time to another are uncorrelated by assuming that the users are randomly dropped for every channel realization.
To evaluate the considered approaches, we first consider $\bar b$ quantization bits for the fixed ADCs. 
Then, the total ADC power of such fixed-ADC system is used for a total ADC power constraint of the revMMSQE-BA as in \cite{choi2017low}.
Then, the total receiver power consumption of the system with revMMSQE-BA is adopted for a power constraint of the proposed-BA system.

In Fig. \ref{fig:REE}, the sum spectral efficiency and the energy efficiency are simulated over different $\bar b$.
In Fig. \ref{fig:REE}(a), the proposed BA shows the highest sum spectral efficiency in the small $\bar b$ range ($\bar b \leq 3$) and converges to the infinite-resolution system faster than the other cases.
In Fig. \ref{fig:REE}(b), the proposed BA also provides the highest energy efficiency in the small $\bar b$ range with the highest rate.
This indicates that the proposed BA algorithm eliminates most of the quantization distortion requiring the minimum power consumption.
Accordingly, we can employ existing state-of-the-art digital beamformers to the power-constrained system when using the proposed BA algorithm as it makes the quantization error negligible in the low-resolution regime.
In the large $\bar b$ range ($\bar b >7$), the energy efficiency of the proposed BA  and revMMSQE-BA becomes lower than that of the fixed ADC due to the dissipation of power consumption in resolution switching. 
The proposed BA, however, achieves the sum spectral efficiency close to the upper bound in the small $\bar b$ range ($\bar b \leq 3$), and thus, only the energy efficiency in $\bar b \leq 3$ is worth to consider.

\section{Conclusion}
\label{sec:Con}

In this paper, we develop an algorithm which minimizes quantization error of desired signals by allocating an optimal number of quantization bits under constrained total receiver power for mmWave communications.
Proposing an off-line training and modeling approach to estimate the average resolution switching power, we convert the quantization error minimization problem to a convex optimization problem to resolve non-linearities in the problem. 
Then, we derive a closed-form solution 
and utilize the solution in the binary search to find the optimal number of activated RF chains which leads to the smallest quantization error, providing a near optimal ADC bit configuration.
Simulation results validate the spectral and energy efficiency of the proposed algorithm.
In particular, existing digital beamforming techniques for hybrid mmWave receivers can be applied to the considered system with the proposed BA algorithm as it makes the quantization error negligible in the low-resolution regime.
Therefore, this paper provide a novel solution to establish spectrum- and energy-efficient systems for mmWave receivers by adapting quantization resolution.
\bibliographystyle{IEEEtran}
\bibliography{Globecom17.bib}
\end{document}